\PassOptionsToPackage{usenames,dvipsnames}{xcolor}

\documentclass[a4paper,11pt]{article}

\usepackage[utf8x]{inputenc}

\usepackage{multirow}
\usepackage{latexsym}

\usepackage{enumitem}

\usepackage{amsmath,amssymb,amsthm}
\usepackage{mathrsfs}
\usepackage{xfrac}
\usepackage{xspace}

\usepackage{graphicx}
\usepackage{epic,eepic}
\usepackage{pstricks, pst-text}

\usepackage{booktabs} % for much better looking tables
\usepackage{array} % for better arrays (eg matrices) in maths

\usepackage{subfigure} % make it possible to include more than one captioned figure/table in a single float
\usepackage{pdfpages}
\usepackage{lastpage}
\usepackage[ruled,vlined,linesnumbered,resetcount]{algorithm2e}
\usepackage{multicol}

\usepackage{cite}
\usepackage[bookmarks=false,colorlinks,breaklinks]{hyperref}
\hypersetup{linkcolor=blue,citecolor=blue,filecolor=blue,urlcolor=blue}

\usepackage{cleveref}

\usepackage{color, colortbl}

\usepackage{tikz}
\usepackage{xkeyval}
\usepackage[textsize=small]{todonotes}

\usepackage{Terminology-GT}
\usepackage{preamblenew}

\usepackage{vmargin}
\setmarginsrb{1in}{1in}{1in}{1in}{0mm}{0mm}{0mm}{7mm}

\usepackage{calc,ifthen}
\newcounter{hours}
\newcounter{minutes}
\newcommand{\Printtime}{\setcounter{hours}{\time/60}%
\setcounter{minutes}{\time-\value{hours}*60}%
\thehours:%
\ifthenelse{\value{minutes}<10}{0}{}\theminutes}

\begin{document}

\title{Stable Nash Equilibria in the Gale-Shapley Matching Game}
\author{Sushmita Gupta$^{1}$ and Kazuo Iwama$^{1}$ and Shuichi Miyazaki$^{2}$}
\author{Sushmita Gupta\thanks{Kyoto University, Japan. \newline
Emails: \texttt{sushmita.gupta@gmail.com, iwama@kuis.kyoto-u.ac.jp, shuichi@media.kyoto-u.ac.jp}}
\addtocounter{footnote}{-1}  
\and  Kazuo Iwama\footnotemark\addtocounter{footnote}{-1} \and Shuichi Miyazaki\footnotemark }

%\thanks{Graduate School of Informatics, Kyoto University}
%\email{sushmita.gupta@gmail.com, iwama@kuis.kyoto-u.ac.jp}
%\and
%\thanks{Academic Center for Computing and Media Studies, Kyoto University}
%\email{shuichi@media.kyoto-u.ac.jp}

\date{}
\maketitle

\begin{abstract}
In this article we study the stable marriage game induced by the men-proposing Gale-Shapley algorithm. 
Our setting is standard: all the lists are complete and the matching
mechanism is the men-proposing Gale-Shapley algorithm. It is well known that in this setting, men cannot
cheat, but women can. In fact, Teo, Sethuraman and Tan \cite{TST01}, show that there is a polynomial time
algorithm to obtain, for a given strategy (the set of all lists) $Q$ and
a woman $w$, the best partner attainable by changing her list. 
However, what if the resulting matching is not stable with respect to $Q$?
Obviously, such a matching would be vulnerable to further manipulation, but is not
mentioned in \cite{TST01}. In this paper, we consider (safe) manipulation that
implies a stable matching in a most general setting. Specifically, our
goal is to decide for a given $Q$, if $w$ can manipulate her list to obtain a strictly better partner with respect to the true strategy $P$ (which may be different from $Q$), and also the outcome is a stable matching for $P$. 
\end{abstract}

\section{Introduction}
Matchings under preferences is an extensively studied area of
theoretical and empirical research that has a wide range of
applications in economics and social sciences.  The most popular and
standard problem in this field is the {\it stable marriage problem} (SMP),
introduced by Gale and Shapley \cite{GS62}, where we have two parties; a
set of men and a set of women.  Each man has a preference list that
orders women according to his preference, and similarly each
woman also has a preference ordering for the men. The goal is to find a stable matching, a matching without
any {\em blocking pair}, where a blocking pair is a pair of a man and a woman who prefer each other to their current partners.

The goal of each player is of course to obtain a good partner (to steal a good partner from colleagues).  Thus, the stable marriage is a typical game played by selfish players and its game-theoretic aspects have been studied for a long time (see Roth and Sotomayor~\cite{RS90b}). For instance, it is well known that if we use the men-proposing Gale-Shapley algorithm (GS-M, in short) as the matching mechanism, then men do not benefit by cheating, \ie., any change of a single list on the men's side does not give him a strictly better partner. Also, if we allow individuals to truncate their lists (\ie delete one or more persons from the list), then there are several (relatively easy) ways to manipulate  to obtain a better outcome. 

For the special case, when all lists are complete, we have known about instances of SMP (see page 65 of \cite{GI89b}, for example), where a woman can improve her outcome by permuting her true preference list, while all others use their true lists. But little else was known about the property of {\it permutation strategies}, such as the time complexity of finding such a list.  It is due to the work of Teo \etalcite{TST01} that this question was resolved. They proved that there is a poly-time algorithm to obtain, for a given set of complete preference lists for all men and women, and a particular woman $w$, the best
partner that $w$ can attain by using one of the $n!$ permutations of her true preference list. As a corollary, it is also
possible to decide in polynomial time if the set of true preference lists
is a Nash equilibrium. Thus, \cite{TST01} was a nice start towards progress
in this topic and their proof is also nontrivial.  However, their
setting that there is no distinction between the ``true'' preference
list of a player and the actual preference list submitted by him/her, is not standard in game theory. This is supported by theory, by the result that there is no mechanism in which the best strategy for all the agents is to tell the truth \cite[Thm 4.4]{RS90b}. For GS-M it is in the best interest for the men to tell the truth, but that is not the case for the women. Hence, we have to consider the possibility that true preference is different from stated preference. 

Recall that in the general setting of game theory, we have a set of
players and their strategies.  The outcome of each player $i$ for his/her
strategy $s_i$, is determined by a function of $s_i$ and all the
other players' strategies.  In the case of stable marriage, the
strategy of each player is his/her preference list, and that can be one of $n!$
different ones and the player's outcome is determined by final matched partner
at the end of the matching algorithm, (GS-M in our case). The true preference list is necessary for assessing an agent's profit or loss; 
the actual stated strategy (submitted list) may of course be very different from
the true list. In other words, a stated list of a woman $w$ and
her true lists are effectively independent; the latter is necessary for computing the quality (\ie benefit, if any) of the outcome for $w$, and as it turns out, has another important role. 

Consider the strategies depicted in Figure \ref{fig:Q1-Q4}: We have five men $\Sb{1,\ldots, 5}$, five
women $\Sb{a,\ldots, e}$, and their current preference lists, given by $Q_{1}$.  Recall that
each person has his/her true preference list, and that is depicted by $P$ in the
table. We want to see how vulnerable $Q_{1}$ is against manipulation. GS-M applied to $Q_{1}$ yields $\mu_{Q_{1}}=\Sb{ (1,a), (2,b), (3,c), (4,d), (5,e)}$, which is
stable with respect to $P$ (called $P$-stable), since $\M_{P}=\M_{Q_{1}}$. However, it is not 
stable in the Nash equilibrium sense, since $d$ can improve her
outcome by changing her list to $Q_{2}(d)$, resulting in $\M_{Q_{2}}=\Sb{ (1,a),
(2,b), (3,c), (4,e), (5,d)}$.  Her new partner $5$ is better than $4$ in terms of $P(d)$ and the matching is again $P$-stable. It is noteworthy that $\M_{Q_{2}}$ is the women-optimal $P$-stable matching, hence it admits no \Ps manipulation strategies. 

But then again, another woman $a$ gets into the game, against $Q_{2}$, and using $Q_{3}(a)$. Her new partner is $2$, who is ranked first in her true list.Unfortunately, though, $\mu_{Q_{3}}$ is not $P$-stable, with 
blocking pairs $(2,b)$ and $(3,b)$. This leaves open the possibility of manipulation by $b$, given by $Q_{4}(b)$, whereby $b$ is matched to $2$, and $a$ to $3$. Thus, $a$ actually has a lower payoff in $Q_{4}$ than $Q_{2}$, thus rendering her manipulating of $Q_{3}$ counter-productive.  This is why we posit that manipulation strategies that are \Ps add a layer of security for the manipulation, since unstable manipulation will always leave open the possibility of others to manipulate, specifically those women who are in a $P$-blocking pair. 

%Thus, we need to consider two stabilities, one for the change of player's mind (which is a usual concept of Nash Equilibrium) and the other for the stability of the matching itself.

%If we assume that only $P$-stable matchings are acceptable as safe result, this $\M_{Q3}$ does not give any profit to woman $a$.  In other %words, her attempt of getting better off once seems successful but is actually not.  

\begin{figure}
\[\begin{array}{|c| ccccc |c| ccccc| ccccc|}
\hline
{\rm Men}& &\multicolumn{3}{c}{P(M)}  & & {\rm Women} & &\multicolumn{3}{c}{P(W)} &  &&\multicolumn{3}{c}{Q_{1}(W)}& \\
\hline
1:& a& b& d& e& c&  a: & 2& \underline{1}& 4& 5& 3&    2& \underline{1}& 4& 5& 3\\
2:& b& a& d& e& c&  b: & \underline{2} & 3& 1& 4& 5&    1& \underline{2}& 4& 5& 3\\
3:& a& b& d& e& c&  c: & 1& 2& 4& 5& \underline{3}&    1& 2& 4& 5& \underline{3}\\
4:& d& e& a& b& c&  d: & 5& \underline{4} & 1& 2& 3&    5& \underline{4}& 1& 2& 3\\
5:& e& d& a& b& c&  e: & 4& \underline{5}& 1& 2& 3&    4& \underline{5}& 1& 2& 3\\
\hline
\end{array}\]

\small{True strategy $P=(P(M),P(W))$ and stated strategy $Q_{1}=(P(M), Q_{1}(W)$.}
\small{The $M$-optimal matching $\M_{P}=\M_{Q_{1}}=\Sb{(1,a), (2,b), (3,c), (4,d), (5,e)}$, thus, $\M_{Q_{1}}$ is \Ps.} 

%\label{fig:P-Q1}
%\caption{True preference $P=(P(M),P(W))$ and a \Ps strategy $Q_{1}=(P(M), Q_{1}(W)$.}
%\end{figure}

%\begin{figure}
\[\begin{array}{|c| ccccc |c| ccccc|ccccc|ccccc|}
\hline
{\rm Men}& &\multicolumn{3}{c}{P(M)}  & & {\rm Women} & &\multicolumn{3}{c}{Q_{2}(W)} &~ & ~&\multicolumn{3}{c}{Q_{3}(W)}& &&\multicolumn{3}{c}{Q_{4}(W)} & \\
\hline
1:& a& b& d& e& c& a: & 2& \underline{1}& 4& 5& 3& \underline{2}& 3& 1& 4& 5 & 2& \underline{3}& 1& 4& 5\\
2:& b& a& d& e& c& b: & 1& \underline{2} & 4& 5& 3 & \underline{1}& 2& 4& 5& 3 & \underline{2}& 3& 1& 4& 5\\
3:& a& b& d& e& c& c: & 1& 2& 4& 5& \underline{3} & 1& 2& 4& 5& \underline{3} & \underline{1}& 2& 4& 5& 3\\
4:& d& e& a& b& c& d: & \underline{5}& 3& 4& 1& 2 & \underline{5}& 3& 4& 1& 2 &5& 3& \underline{4}& 1& 2\\
5:& e& d& a& b& c& e: & \underline{4}& 5& 1& 2& 3 & \underline{4}& 5& 1& 2& 3 &4& \underline{5}& 1& 2& 3\\
\hline
\end{array}\]
\small{The $P$-blocking pairs in $\M_{Q_{3}}$ are $(2,b)$ and $(3,b)$. Manipulation by $b$ gives $Q_{4}=(Q_{3}(-b), P(b))$, resulting in $\M_{Q_{4}}=\Sb{(1,c), (2,b), (3,a), (4,d), (5,e)}$. }
\label{fig:Q1-Q4}
\caption{Manipulation by $d$, followed by $a$, and then $b$.}
\end{figure}

\paragraph{{\bf Our Contributions:}} Motivated by these strategic considerations,
we introduce a {\it stable \Nash}, which is a relaxation of a \Nash, but
nonetheless sufficiently limits an agent's incentive to manipulate the
market. Suppose that the stable matching game is induced by the men-optimal stable matching mechanism.  
For a given strategy $Q$ and a true strategy $P$, $Q$ is said
to be a {\it P-stable \Nash} (in this game) if and only if there is no woman $w$ such
that (i) $w$ can manipulate her list, (while all others retain their $Q$-lists) and obtain a better partner with
respect to $P(w)$, and (ii) the resulting matching is $P$-stable.  Note that there
are several cases that people may try to improve their outcome but actually fail under (i) and (ii). 
For instance, (a) $w$ manipulates her list, and indeed gets matched to $x$, who is better
with respect to $Q$, but not with respect to $P$, and (b)
$w$ manipulates her list, gets indeed a better partner $x$ with respect
to $P$, but the resulting matching is not $P$-stable, for instance $Q_{3}$, in our example. 

Our main result in this paper is the existence of a polynomial time
algorithm that determines if $Q$ is a $P$-stable \Nash.  Our algorithm
has two stages.  In the first stage we obtain triples $(w, x, Q(x;w))$, 
where $w$ is a woman, $x (\neq \mu_{Q}[w])$ is a man, and $Q(x;w)$ is a
list of $w$ such that $\mu_{(Q(-w), Q(x;w))}$ matches $x$ to $w$. Here, the preference lists $(Q(-w), Q(w,x))$, is 
such that only $w$'s list is changed to $Q(x;w)$, while each $a\in M\cup W\sm \{w\}$ uses $Q(a)$ (formal
definitions in \Cref{Sec:Prelim}).  It should be noted that
there may be many different lists of $w$ that if used instead, would also match $w$ to $x$, but this
algorithm outputs just one of them.  The algorithm itself is basically
the same as \cite{TST01}, but the analysis is completely new.  Recall that \cite{TST01} only guarantees that
it outputs a similar triple $(w, x, Q(x;w))$, where $x$ is the {\em best}
possible partner of $w$ (with respect to $Q(w)$). Our arguments are completely new, proving the
more general statement, that {\em all} partners attainable by GS-M can be found using the algorithm.

The second stage is to decide, for each triple $(w, x, Q(x;w))$ obtained
in the first stage, whether or not there exists a list for $w$, say
$Q'(x;w)$, which if used instead of $Q(w)$ would result in $w$ to $x$, with the condition that $\mu_{(Q(-w),Q'(x;w))}$ is
$P$-stable. List $Q(x;w)$ is an obvious candidate for such a $Q'(x;w)$, but
there may be (exponentially) many others, and it is possible that $Q'(x;w)$ results in a \Ps matching but $Q(x;w)$ does not. 
For this, we prove a key lemma that is interesting in its own right. Specifically, we show that $\mu_{(Q(-w),Q'(x;w))}$
is unique regardless of the list $Q'(x;w)$. Therefore, to detect if $Q$ is \Ps \Nash, it
suffices to check if $\mu_{(Q(-w),Q(x;w))}$ is $P$-stable for all the (at
most $n^2$) triples obtained in the first stage.

\paragraph{{\bf Related Work:}}%For the stable marriage problem, ideally there should exist a matching mechanism that makes stating the true preferences the dominant strategy for each agent, a mechanism such as this is called a {\it strategy-proof mechanism}. Roth \cite{RS90b} showed that for the men-proposing Gale Shapley algorithm (GS-M, in short), telling the truth is the dominant strategy for all men. But, \sout{as stated earlier,} there is no mechanism in which the dominant strategy for all the agents to tell the truth.

We denote the true preferences of $(M,W)$ by $P$, and for an individual agent $a\in M \cup W$, the preference list of $a$ by $P(a)$. The $M$-optimal and $W$-optimal matching \wrto a given strategy $Q$ is denoted by $\M_{Q}$. % and $\Mo{W}$, respectively. 

For a stated strategy $Q$, Dubins and Freedman \cite{DF81}, proved that there is no coalition $C$ of men, who have a \MS $P'=(P(-C), P'(C))$ so that the outcome $P'$-stable, and is strictly better than $\M_{P}$, in $P(m)$, for each $m\in C$.  Demange, \etalcite{DGS87} extend this result to include women in the coalition $C$, showing that there is no $P'$-stable matching $\M'$ such that every agent $a\in C$, prefers $\M'[a]$, (in $P(a)$) to his/her partner in {\it every} $P$-stable matching. 

For truncation strategies it was shown by Gale and Sotomayor \cite{GS85} that if there are at least two \Ps matchings, then there is a woman $w$ who has a  unilateral \MS $Q'\in \C{Q}_{w}$ that gives a strictly better outcome than $\M_{P}$, in $P(w)$. If $C=W$, then there is a truncation strategy $P'=(P(M),P'(W))$, such that $\M_{P'}$, is the women-optimal matching for $P$. Considerable work on truncation strategies have been undertaken, see \cite{Ehlers08, RR99} for motivations and applications. In fact, up until the late 1980s, analyses of manipulation strategies of women centred almost exclusively around truncation strategies. 

As stated earlier, (to the best of our knowledge) Teo \etalcite{TST01} was the first to consider permutations of a agent's true lists as a type of \MS. 
They gave a polynomial time algorithm that for a given $w\in W$, computes the best manipulated partner (in $P(w)$) that $w$ can attain using GS-M on any strategy $(P(-w), P'(w))\in \C{P}_{w}$. 

%Teo \etalcite{TST01} considered  $\C{P}_{W}=\Sb{(P(-w), P'(w))\mid w\in W}$, complete list unilateral manipulation strategies of women that involves a woman $w$ misstating her preference, while all others state their true preference. 

Immorlica and Mahidian \cite{IM05}, allowing the men's preference lists to have ties but bounded by a constant, and drawn from an arbitrary probability distribution, while the women's lists are arbitrary and complete; show that with high probability, truthfulness is the best strategy for an any agent, assuming everybody is being truthful as well.

Kobayashi and Matsui \cite{KM10j} consider the possibility that a coalition $C$ of women have a \MS $P'=(P(M),P'(W))$, containing complete lists,  such that $\M_{P'}$ yields specific partnerships for the members of $C$. The situation manifests in two specific forms, depending on the nature of the input. In the first case, the input consists of the complete lists of all men, a partial matching (some agents may be unmatched) $\M'$, and complete lists of the subset of women who are unmatched in $\M'$, denoted by $W\sm C$. The problem is to test whether there exists a permutation strategy for each woman in $C$, such that for the combined strategy $P'=(P(-C), P'(C))$, $\M_{P'}$ is a perfect matching that extends $\M'$. In the second case, the input consists of the lists of all men, a perfect matching $\M$, and lists for women in $W\sm C$. The problem is to test if there are permutation strategies for the women in $C$, such that for $P'=(P(-C), P'(C))$, matching $\M_{P}'=\M$. %\Ma{Separately, Matsui \cite{ } discussed \eqm .....}\maG{I could not get hold of this paper.}

Pini \etalcite{PRVW11} show that for an arbitrary instance of the stable marriage problem, there is a stable matching mechanism for which it is NP-hard to find a \MS. More recently, Deng, \etalcite{DST15a}, drawing on \cite{KM10j} have discussed the possibility of a coalition of women permuting their true lists, while others state theirs truthufully, so as to produce a matching that is Pareto optimal for the members of the coalition. They also give an \Bo{n^{6}} algorithm to compute a strong \Nash that is strongly Pareto optimal for all the coalition partners.

Our work deals with manipulation by women, in the game induced by GS-M, where truth-telling is the best strategy for men. For a background on manipulation strategies for men, we point the reader to \cite{H06, IIIMN13} and \cite{Manlove}.

\section{Preliminaries}\label{Sec:Prelim}

We always use $M$ to denote the set of $n$ men $\{m_1, m_2, \ldots, m_n\}$
and $W$ the set of $n$ women $\{w_1, w_2, \ldots, w_n\}$.  Our matching
mechanism is always the men-proposing Gale-Shapley algorithm: On the
man's side, a man who is {\em single}, {\em proposes} to the woman who is at the
top of his current list.  At the woman's side, when a woman $w$ receives
a proposal from a man $m$, she {\em accepts} that proposal if she is
single.  Otherwise, if she prefers $m$ to her current partner $m'$ ($m'$
to $m$, resp.), then $w$ accepts $m$ ($m'$, resp.) and {\em rejects}
$m'$ ($m$, resp.).  If $m$ is rejected by $w$, then $m$ becomes single
again and $w$ is removed from the $m$'s list.  The algorithm terminates
when there is no single man (for more details, see \cite{GI89b}).  We
assume that if there are two or more single men, the man with the
smallest index does a proposal, thus making the procedure deterministic.

A {\em strategy} $Q$ is a set of {\em preference lists} (or simply {\em
lists}) of all the men in $M$ and all the women in $W$. For a person $x$
in $M \cup W$, $Q(x)$ denotes the $x$'s list in the strategy $Q$. For a given
strategy $Q$, suppose that {\it only} $w$ changes her list from $Q(w)$ to
$Q'(w)$. We denote the resulting strategy by $Q'=(Q(-w),Q'(w))$, and use $\C{Q}_w$ to denote the family of all such strategy $Q'$.  Note that all lists considered in
this article are complete,\ie, they are permutations of $n$ men or $n$
women. 

Let $\mu$ be a (perfect) matching between $M$ and $W$ and $Q$ be a
strategy.  Then we say that $\mu$ is {\em $Q$-stable} if there is no
{\em $Q$-blocking pair}.  Here $Q$-blocking pair is a pair $(m,w)$ of
a man and a woman such that $m$ prefers, in (terms of the preference ordering in) $Q(m)$, $w$ to his partner
in $\mu$ (denoted by $\mu[m]$) and $w$ prefers, in $Q(w)$, $m$ to
$\mu[w]$. If $w$ strictly prefers $m_1$ to $m_2$, in $Q(w)$, then we denote that as 
$m_1 > m_2$ in $Q(w)$. We use $m_1 \geq m_2$, if $m_1 > m_2$, or $m_{1}=m_{2}$.

For a strategy $Q$, $\mu_Q$ denotes the man-optimal stable matching, computed by 
the Gale-Shapley algorithm.  If a man $m$ proposes to a woman $w$ during
this procedure, then we say that $m$ is {\em active} in $Q(w)$ (formally
speaking we should say $m$ is active in $Q(w)$, during the computation of $\mu_{Q}$, but for the sake of brevity, we will omit strategy $Q$ when it is obvious from the context.)

Recall that a woman $w$ changes her list $Q(w)$ for the purpose of manipulation. For a subset $M' \sse M$, let $I$ be an ordering of men in $M'$. Then, $Q(I;w)$ denotes a permutation of $Q(w)$, where the men in $M'$ are at the front in the order in which they appear in $I$. An ordered (sub)list such as $I$, is called a {\it tuple}, and for any given tuple $I$, we define $Q(I;w)=[I, Q(w)\sm I]$
For example, if $Q(w)=(1,2,3,4,5,6)$ and $I=(5, 2)$, then $Q(I;w) =(5,2,1,3,4,6)$. Now we are ready to introduce our main concept.

We are given a strategy $Q$ and a {\em true} strategy $P$.  Then for a woman $w\in
W$, a strategy $Q' \in \C{Q}_{w}$ is said to be a {\it unilateral \MS}
of $w$, if $\M_{Q'}[w] > \M_{Q}[w]$ in $P(w)$, \ie $w$ strictly
prefers the outcome of $\M_{Q'}$ to $\M_{Q}$, with respect to her true
preference/strategy.  If $\M_{Q'}$ is a \Ps matching, then $Q'$ is said to be
a {\it \Ps \MS} of $w$. A strategy $Q$ is said to be a {\it $P$-stable \Nash} if there does
not exist $w\in W$ having a \Ps unilateral \MS $Q'\in \C{Q}_{w}$.

In this paper, we consider the following problem:

\begin{description}\label{Prob}
\item[Problem:] Given a true preference $P$, and a stated preference $Q$, determine if $Q$ is a $P$-stable \Nash. 
%%\item[Problem 2:] Given a true preference $P$ and a stated preference $Q$, determine if $Q$ is a \Nash.
\end{description}

%\begin{problem}\label{Prob}
%Given a true preference $P$, and a stated preference $Q$, determine if $Q$ is a $P$-stable \Nash.
%\end{problem}

\section{Listing active men}

Now our goal is to design an algorithm that, for two given strategies,
a stated strategy $Q$ and a true strategy $P$, answers if $Q$ is a $P$-stable \Nash.  To
do so, we first design an algorithm that outputs the set $N_w(Q)$ of
all possible partners $m$ of a given fixed woman $w$ such that there
is a (manipulated) strategy $Q'=(Q(-w),Q'(w))$, for which the men-optimal stable matching (\ie men-proposing GS algorithm), will match $w$ to $m$. 
By using this algorithm $n$ times, we can obtain $\Sb{N_{w_1}(Q), \ldots, N_{w_n}(Q)}$. The use of this set to prove our main result is explained in the next section.

%If at least one of them includes a man$m$, whom the corresponding $w_{i}$ prefers to $\M_{Q}[w_{i}]$, in $P(w_{i})$, and the resulting matching is  \Ps, then we can conclude that $w_{i}$ has a \Ps \MS, and so $Q$ is not a $P$-stable \Nash.  This conclusion will be proved in the next section.

See \Cref{Algo:Explore}, which is basically the same as the one given by Teo \etalcite{TST01}:
Suppose that $Q(w)=(1,2,3,4,5,6,7,8)$ and the first proposal comes from
man 5.  Then the algorithm adds 5 to $N$ and call ${\bf Explore}(Q(5;w))$,
namely it executes men-proposing Gale-Shapley algorithm (GS-M, in short) after moving man 5 to the front of
$Q(w)$. In general, procedure {\bf Explore} takes as a parameter, $Q(x,I;w)$, a preference list of $w$. As per our notation, $x$ is at the front of this list, followed by the sublist $I$ and then the rest of the men, thus, defining the strategy $Q'=(Q(-w), Q(x, I;w))$.  ${\bf Explore}(Q(x,I;w))$ executes GS-M for the strategy $Q'$ and
produces the set of men $A$ who propose to $w$ after $x$. Now for each
$y \in A$, we check if $y$ is ``new'' (\ie, not yet in $N$).  If so,
we add $y$ to $N$ and call {\bf Explore} recursively
after moving $y$ to the top of $Q(x,I;w)$, else, we do
nothing. 

Since {\bf Explore} is called only once for each man in $N$, its time
complexity is obviously at most $n \times T({\rm GS})$, where $T({\rm GS})$ is the time
complexity of one execution of GS-M, thus, overall it is $\Bo{n^3}$.  The nontrivial part
is the correctness of the argument, which we shall prove now.

\begin{theorem}\label{Theorem:Explore} For a strategy $Q$ and a woman
$w\in W$, \Cref{Algo:Explore} produces \\$N=\Sb{m\in M\mid \exists
Q'\in \C{Q}_{w},~\textrm{s.t.}~\M_{Q'}[w]=m}$ and for each $m \in N$, a
list $Q(m,I; w)$ such that for some partial list $I$, $m$ is active in $Q(I;w)$.
\end{theorem}

\begin{proof} Let $Q'(w)$ be an arbitrary permutation of $n$ men and
$Q'$ be the strategy $(Q(-w), Q'(w))$.  It is enough to prove if a man
$x\in M$ proposes to $w$ during the computation of $\M_{Q'}$ (\ie $x$
is active in $Q'(w)$), then $x$ is added to $N$ during the execution
of \Cref{Algo:Explore}.

Here we need two new definitions: Suppose that $x_1, x_2, \ldots, x_t$ is
a sequence of proposals received by $w$ during the computation of
$\M_{Q'}$.  Then this sequence is called an {\em active sequence} for
$Q'(w)$, denoted by ${\rm AS'}(w)$.  Also suppose that $y_1,y_2, \ldots, y_s$
is a maximal subsequence of ${\rm AS'}(w)$ such that $y_1=x_1$ and $y_s >
y_{s-1} > \cdots > y_1$ in $Q'(w)$.  Then, this is
called the {\it increasing active subsequence} for $Q'(w)$ and is denoted by
${\rm IAS}'(w)$.  As an example, let $Q'(w)=(1, 2, 3, 4, 5, 6,7,8,9)$ and
${\rm AS}'(w)={5, 6, 3, 4, 2, 8}$.  Then ${\rm IAS}'(w)={5,3,2}$. Now consider a
different list $Q''(w)=(1,2,3,5,9,8,4,6,7)$, thus, $Q'(w)\ne Q''(w)$. 
Then, we can observe that the active sequence and the increasing
active subsequence for $Q''(w)$ are identical to those of $Q'(w)$.
The reason being the following. The lists in $Q'$ and $Q''$
are the same except that of $w$'s, so the first proposal for $w$ must come
from the same man regardless of $w$'s list. Since 5 is accepted by $w$
(both lists are complete), the next proposal should also be from the same man, 6. Now since $5 > 6$ in both $Q'(w)$
and $Q''(w)$, 6 is rejected in both $Q'(w)$ and $Q''(w)$ and
thus, the next proposal must also be same, and so on, we continue.
This observation leads us to the following lemma.

\begin{lemma}\label{L:Lists-AS}For strategies $Q',Q'' \in \C{Q}_{w}$,
let $x_{1}, x_{2}, \ldots, x_{p}$ and $u_{1}, u_{2}, \ldots, u_{q}$,
denote the \ASs for $Q'(w)$ and $Q''(w)$ respectively and let $y_{1},
y_{2}, \ldots, y_{s}$ and $v_{1}, v_{2}, \ldots, v_{t}$ denote the
corresponding increasing active subsequence. Then, the following conditions must hold.

\begin{enumerate}[label=(\alph*)]
\item\label{1st-AM} $x_1=y_1=u_1=v_1$.

%\item\label{rest-AS} If $x_i=u_i$ for all $i$, where $1\leq i\leq l$ for some $l$ ($l\leq p$ and $l\leq q$), and $y_j=v_j$ for all $j$,
%where $y_1, ..., y_j$ and $v_1, ..., v_j$ are increasing \ASs of $x_{1}, x_{2}, \ldots, x_{l}$ and $u_{1}, u_{2}, \ldots, u_{l}$,
%respectively, then $x_{l+1}=u_{l+1}$.  (Note that $y_1, ..., y_j$ and $v_1, ..., v_j$ may be different in general if $x_{1}, x_{2}, \ldots,
%x_{l}$ and $u_{1}, u_{2}, \ldots, u_{l}$ are the same.)\\

\item\label{rest-AS} For an arbitrary $l$ ($l\leq p$ and $l\leq q$), we consider the prefixes of the active sequences up to position $l$
and the prefixes of the corresponding increasing active
subsequences, denoted by $y_1, ..., y_j$ and $v_1, ..., v_j$.  Then, if
$x_i=u_i$ for all $i \leq l$ and $y_k=v_k$ for all $k \leq j$, then
$x_{l+1}=u_{l+1}$.

\end{enumerate}
 \end{lemma}

\begin{proof} By definition, $x_{1}=y_{1}$ and $u_{1}=v_{1}$.
Recall that all the lists in $Q'$ and $Q''$ are the same except those for
$w$.  Furthermore, we use a fixed tie-breaking protocol in the
deterministic GS algorithm. Hence, $x_{1}=u_{1}$, follows directly. 

To prove condition \ref{rest-AS}, let $y_2=x_{i_1+1}, y_3=x_{i_2+1}, \ldots$, and so on. 
Then we can write ${\rm AS'}(w)$ as follows, where $x_2, \ldots,
x_{i_1}, x_{i_1+2},\ldots, x_{i_2}, \ldots $ may be empty.
\[{\rm AS'}(w)=y_1, x_{2}, \ldots, x_{i_1}, y_2, x_{i_1+2}, \ldots, x_{i_{2}},\ldots, y_j, x_{i_{j-1}+2}, \ldots,
x_l, x_{l+1}, \ldots\]

Now one can see that $y_1$ is accepted, all of $x_{2} \ldots, x_{i_1}$
are rejected since they are after $y_1$ in the list by definition.
This continues as $y_2$ is accepted,  $x_{i_1+2}, \ldots,
x_{i_2}$ rejected, and so on.  Now, consider ${\rm AS''}(w)$, depicted below.
\[{\rm AS''}(w)=v_1, u_{2}, \ldots, u_{i_1}, v_2, u_{i_1+2}, \ldots, u_{i_{2}}, \ldots, v_j, u_{i_{j-1}+2}, \ldots,
u_l, u_{l+1}, \ldots\]

By the assumption, these two sequences are identical up to
position $l$, so acceptance or rejection for each proposal follows identically, as discussed above. Therefore, the configuration (see below) of GS-M for $Q'$ at the moment when $x_l$ proposes to $w$ and the {\it configuration}
for $Q''$ when $u_l$ proposes to $w$ are exactly the same.  A
configuration consists of (i) the lists of all men (recall that some
entries are removed when proposals are rejected), (ii) the set of
single men, and (iii) the current temporal partner of each woman.
(Formally this should be shown by induction, but it is straight forward
and may be omitted).  Also the acceptance/rejection for $x_l$ and
$u_l$ is the same.  Thus in either case, the execution of the
(deterministic) GS-M is exactly the same for $Q'$ and
$Q''$ until $w$ receives proposal from $x_{l+1}$ and $u_{l+1}$, respectively. Hence, $x_{l+1}$ and
$u_{l+1}$, should be equal and the lemma is proved. 

%Therefore, at the moment when $u_l$ is now proposing to
%$w$, (i) the lists of all men (recall that some entries are removed
%when proposals are rejected) and (ii) the set of single men are and
%(iii) the current temporal partner of each woman (namely the
%configuration of the Gale-Shapley procedure) are exactly the
%same between $Q'$ and $Q''$ (formally this should be shown by
%induction, it is straightforward and may be omitted).  Also the
%acceptance/rejection for $u_l$ is the same.  Thus in either case, the
%execution of the (deterministic) Gale-Shapley is exactly the same
%between $Q'$ and $Q''$ until the propose comes back to $w$.  Thus
%$x_{l+1}$ and $u_{l+1}$ should be equal and the lemma is proved. 
\end{proof}

Now let us look at the execution sequence of \Cref{Algo:Explore} while
comparing it with the execution sequence of GS-M on $Q'$.
We assume that the active sequence and increasingly active sequence
for $Q'$ are ${\rm AS'}(w)=x_{1}, x_{2}, \ldots, x_{p}$ and ${\rm IAS'}(w)= y_{1}, y_{2}, \ldots, y_{s}$, respectively.  By \Cref{L:Lists-AS}, the first proposal
to $w$ is always $y_1$, so the algorithm starts with $\mathbf{Explore}(Q(y_1;w))$ (we simply say the algorithm invokes $Q(y_1;w)$),
and $N=\{y_1\}$, at the very beginning.

Now we note that it is quite easy to see that the active sequence for $Q(y_1;w)$ should be
$y_1,x_2,\ldots, x_{i_1},\ldots $, \ie it should be identical to that of $Q'(w)$ up
to the position $i_1$. The reason is as follows. We already know the first
active man is always $y_1$ and that is also the first symbol in the
increasing active sequence of both. Thus we can use \Cref{L:Lists-AS},
to conclude that the second symbols should also be the same.  Since, $x_2$
is not in ${\rm IAS'}(w)$, meaning that it is rejected, which is also the same
in $Q(y_1;w)$ since $y_1$ is at the top of the sequence.  Thus the
third symbol is the same in both and so on up to position $i_1$.
Then the next symbol in ${\rm AS'}(w)$ is $y_2$ and its that is also active in
$Q(y_1;w)$, meaning $Q(y_2,y_1;w)$ is invoked by the algorithm.  (The
algorithm also invokes $Q(x_2,y_1;w)$, $Q(x_3,y_1;w), \ldots, Q(x_{i_1},y_1;w)$ also, but these are not important for us at this
moment).

We again consider the active sequence for $Q(y_2,y_1;w)$ and by the
same argument presented earlier, we can conclude that it is identical to
${\rm AS'}(w)$ up to position $i_2$ and so $y_{3}$ is found to be an active man. Hence, $Q(y_3,y_2,y_1;w)$ is invoked if $y_{3}$ was not already present in $N$.  

Continuing like this, we note that if $Q(y_s, y_{s-1}, \ldots ,y_1;w)$ is
invoked, then we are done since its active sequence is identical
to that of $Q'(w)$.  However, this case happens only if each $y_{i}$, $2\leq i\leq s$, is a brand new active man found during the invocation of $Q(y_{i-1},\ldots, y_{1};w)$. If one them is not new then the subsequent lists are not invoked, and yet, we are assured due to \Cref{L:Active-Men} that \Cref{Algo:Explore} will detect all the active men in $Q'(w)$.

%======{\bf Motivation for \Cref{L:Active-Men}}=======\\

\Cref{L:Active-Men} is rather surprising and maybe of independent interest. For two lists $Q'(w)$ and $Q''(w)$, that are distinct and arbitrary orderings on men, we assume nothing about the execution of GS-M on the two lists except that a particular man $x$ is active in either list. Yet, we are able to show that a man who proposes to $w$ after $x$ when $Q'(x;w)$ is used, must also propose when $Q''(x;w)$ is used. This result allows us to focus solely on active men that have been discovered in the current invocation of $\mathbf{Explore}$, thereby restricting the number of recursion steps to \Bo{n}.  

\begin{lemma}\label{L:Active-Men}
For strategies $Q'$ and $Q''$ in $\C{Q}_w$, suppose that $x$ is active in both $Q'(w)$ and $Q''(w)$. Then a man who is active in
$Q'(x;w)$ after $x$, is also active in $Q''(x;w)$.
\end{lemma}

\begin{proof}Let $y$ denote a man who is active in $Q'(x;w)$ after $x$. Consider the following strategies, \[\textrm{$Q_{y}=(Q'(-w), Q'(y,x;w))$ and $Q_{x}=(Q''(-w),Q''(x;w))$,}\]
and the $M$-optimal matchings denoted by $\M_{1}$ and $\M_{2}$, respectively; \ie $\M_{1}[w]=y$ and $\M_{2}[w]=x$.

The lists are complete, and $Q(y)=P(y)$, so either $\M_{2}[y] > w$, or
$w> \M_{2}[y]$ in $P(y)$.  If $\M_{2}[y]>w$
then matching $\M_{2}$ is $Q_{y}$-stable. Thus, both $\M_{1}$ and
$\M_{2}$ are $Q_{y}$-stable, with $\M_{1}$ being the men-optimal matching. But,
then $w=\M_{1}[y]\geq \M_{2}[y]$ in $P(y)$, contradicting 
the initial assumption that $\M_{2}[y]> w$.
Therefore, we conclude that $w> \M_{2}[y]$ in $P(y)$,
implying that $y$ should propose to $w$ before $\M_{2}[y]$ or that $y$
is active in $Q''(x;w)$. Hence, the proof is complete. 
\end{proof}

The next lemma will complete the proof. Note that we are still
using ${\rm AS'}(w)$, as given in the proof of \Cref{L:Lists-AS}, and using that to define the following lists for $w$, %similar to the above good case.
\[\te{$Q_{1}(w)=Q_{1}(y_{1};w)$, and $Q_{j+1}(w)=[y_{j+1}, Q_{j}(w)\sm \{y_{j}\}]$, for $1\leq j\leq s-1$. }\]

\begin{lemma}\label{L:Induction}For each $j$, $1\leq j\leq s$, the algorithm invokes $Q(y_j,I;w)$ for some tuple $I$ and each man in $\Sb{x_{i_{j-1}+2}, \ldots, x_{i_{j}}, y_{j+1}}$ is active in it.
\end{lemma}

\begin{proof}By induction.  The base case is given by $y_{1}$, for which it has been shown earlier
that $Q(y_{1};w)$ is invoked at the beginning and every man in
$\Sb{x_2,\ldots, x_{i_1}, y_2}$ is active in $Q(y_{1};w)$ after
$y_1$. Thus, the base case has been proved.

Suppose that the induction hypothesis holds for $y_{t}$, where $t\leq s-2$, \ie for some tuple $I$, $Q(y_t, I;w)$ is invoked, and each man in $\Sb{x_{i_{t-1}+2}, \ldots, x_{i_{t}}, y_{t+1}}$ is active in it. We will complete the proof by showing that the hypothesis holds for $y_{t+1}$.
If $y_{t+1}$ is ``new'', \ie it is added to $N$ during the invocation of $Q(y_{t},I;w)$, then $Q(y_{t+1}, y_t, I;w)$, is invoked subsequently.
Using the fact that $y_{t+1}$ is active in both $Q_{t+1}(w)$ and $Q(y_{t},I;w)$, and all the men in $\Sb{x_{i_{t}+2}, \ldots, y_{t+2}}$ are active in $Q_{t+1}(w)$ after $y_{t+1}$, \Cref{L:Active-Men} applied to each of them, implies that they are also active in $Q(y_{t+1}, y_t, I;w)$. Hence, for this case, the hypothesis is proved for $y_{t+1}$.

If $y_{t+1}$ is already in $N$ when $Q(y_t, I;w)$ is invoked, then
for some tuple $I'$, $y_{t+1}$ should have been added to $N$, during the invocation of $Q(I';w)$.  Thus, $Q(y_{t+1}, I';w)$ would have been invoked prior to $Q(y_{t},I;w)$. Using the same argument (on $Q_{t+1}(w)$ and $Q(y_{t+1},I';w)$) that we used for the earlier case, we conclude that even for this case, the hypothesis holds for $y_{t+1}$. 
\end{proof}

Thus, we have shown that all the men in ${\rm AS'}(w)$ are active somewhere during the execution of the algorithm and thus, all are present in $N$ at the end of the execution. This completes the proof of \Cref{Theorem:Explore}. 
\end{proof}

\LinesNumberedHidden
\begin{algorithm}[ht]
\KwIn{Strategy $Q$, and a woman $w\in W$}
\KwOut{Sets $N=\{m \in M \mid \exists ~Q'\in \C{Q}_{w}, ~\M_{Q'}[w]=m\}$, and $L=\{Q(m,I;w)~|~ m\in N,~Q'\!=\!(Q(-w),Q(m,I;w)),~ \M_{Q'}[w]=m\}$}

\vspace{0.25cm}

\ShowLn Let $x_{1}$ be the first active man in $Q(w)$;

\ShowLn Let $N \addto \{x_{1}\}$ and $L \addto \{Q(x_{1};w)\}$;

\ShowLn$\mathbf{Explore}(Q(x_{1};w))$\;
\vspace{0.5cm}

Procedure $\mathbf{Explore}(Q(x;w))$\\

\vspace{0.25cm}

Let $A \addto \Sb{ \textrm{men who are active in $Q(x;w)$ after $x$}}$\;

\ForEach{$y \in A\sm N$}{

$N \addto N \cup \{y\}$ and $L \addto L \cup \{Q(y,x;w)\}$\;

${\bf Explore}(Q(y,x; w))$\;
}

\Return{$(N, L)$}\;

\caption{$\A(Q,w)$, \cite{TST01} }
\label{Algo:Explore}
\end{algorithm}

\section{Algorithm for $P$-Stable Nash Equilibrium}

In this section, we consider the problem of deciding,
for a given true preference $P$, and stated preference $Q$, whether
$Q$ is $P$-stable \Nash.  We show that this problem is solvable in time $\Bo{n^{4}}$. \Cref{Algo:Ps-NE} uses \Cref{Algo:Explore} as a
subroutine.

%\LinesNumberedHidden
\begin{algorithm}[h]
\KwIn{True strategy $P$, stated strategy $Q$, and the set of women $W$.}
\KwOut{Answers ``Yes'',  if $Q$ is a \Ps \Nash, else ``No''. }

\vspace{0.5cm}

\ForEach{$w\in W$}{

Run \Cref{Algo:Explore} to obtain $N_{w}(Q)$ and $L_{w}(Q)$;

Let $N'\addto \Sb{m \in N_{w}(Q)~\textrm{s.t. $w$ prefers $m$ to $\mu_{Q}[w]$ \wrto $P(w)$}}$;

%\If{$N'= \emptyset$}{
%Output ``$Q$ is a \Nash'' \;
%\Return{``Yes''}\;
%}
%\Else{
%Output ``$Q$ is not a \Nash'' \;

\ForEach{$m \in N'$}{

Let $Q(m, I;w) \in L_{w}(Q)$ be the list that yields $(m,w)$ as a matched pair;

Compute $\mu$, the men-optimal stable matching for $(Q(-w), Q(m, I;w))$;

%Compute the men-optimal stable matching $\mu$ with respect to $(Q(-w), Q(m;w))$ \;
 
\If{$\mu$ is $P$-stable}{
\Return{``No''}\;
} 
}
}

\Return{``Yes''}\;

\caption{Algorithm for testing \Ps \Nash}
\label{Algo:Ps-NE}
\end{algorithm}

To show the correctness of \Cref{Algo:Ps-NE}, the following lemma plays a key
role.

\begin{lemma}\label{Lemma:unique}Suppose that $Q$ is an arbitrary (unilateral) manipulation strategy of woman
$w$. Then for any (fixed) man $\overline{m}$, a $Q$-stable matching (if
any) that matches $\overline{m}$ to $w$, is unique. %is same as the one computed by \Cref{Algo:Explore}. %\Blue{unique} 

\end{lemma}

\begin{proof}Let $Q$ be an arbitrary manipulation strategy by which $w$ attains $\overline{m}$, and let $\M$ denote the men-optimal stable matching for $Q$.
  
% \Blue{is $P$-stable} and $\mu[w]=\overline{m}$, where $\mu$ is the man-optimal
% stable matching for $Q$.  Our goal is to show that $\mu=\mu^{*}$.

\Cref{Algo:Explore} computes a list $Q(\overline{m},I;w)$, such that $w$ attains $\overline{m}$ by the strategy $Q^{*}=(Q(-w), Q(\overline{m},I;w))$.  Note that $\overline{m}$ appears at the front of the list $Q^{*}(w)=Q(\overline{m},I;w)$.  Let $\mu^{*}$ be the men-optimal stable matching for $Q^{*}$.  Our goal is to show that $\mu=\mu^{*}$.

%\begin{xclaim}
\begin{claim}
\label{claim:1}
For each $m$, $\mu^{*}[m] \geq \mu[m]$, in $Q(m)$.
\end{claim}

\begin{proof}We begin by showing that $\mu$ is $Q^{*}$-stable.  Note that $\mu$ is
$Q$-stable, and $Q$ and $Q^{*}$ differ only in $w$'s 
list.  Hence, if there is a $Q^{*}$-blocking pair in $\mu$, then it
must contain $w$. However, this is impossible since $w$ is matched with
$\overline{m}$, who is at the front of the list $Q^{*}(w)$. Therefore, 
$\M$ must be $Q^{*}$-stable.

Since $\M^{*}$ is the man-optimal stable matching for $Q^{*}$ and $\mu$ is a $Q^{*}$-stable
matching, consequently, for each man $m\in M$ : $\M^{*}[m] \geq \M[m]$, in $Q^{*}(m)$. Since $Q^{*}(m)=Q(m)$, for each man $m$, the claim is proved. 

\end{proof}

\begin{claim}\label{claim:2}
For each $m$, $\M[m] \geq \M^{*}[m]$, in $Q(m)$.
\end{claim}

\begin{proof}
We will show that $\mu^{*}$ is $Q$-stable.  Suppose that it is not.  Then there is
$Q$-blocking pair in $\mu^{*}$, and it includes $w$ for the same
reason as in the proof of \Cref{claim:1}. Let $(m', w)$ denote a
$Q$-blocking pair.  Then $w>\M^{*}[m']$, in $Q(m')$, and $m' >\M^{*}[w]$, in $Q(w)$. 

By \Cref{claim:1}, $m'$ prefers $w$ to $\M[m']$, and recall that $\M^{*}[w] = \M[w]=\overline{m}$. Hence, $(m', w)$ is a $Q$-blocking pair in $\M$, a
contradiction. Again, for the same reason as in the proof of \Cref{claim:1}, we can
conclude that $\M[m] \geq \M^{*}[m]$, in $Q(m)$, for each man $m$.

\end{proof}

By Claims \ref{claim:1} and \ref{claim:2}, $\M[m] = \M^{*}[m]$ for
 each man $m$. Thus, $\M=\M^{*}$, completing the proof of
 \Cref{Lemma:unique}. 
\end{proof}

\begin{theorem}
\Cref{Algo:Ps-NE} solves our Problem in $O(n^{4})$ time.
\end{theorem}

\begin{proof}For each woman, \Cref{Algo:Ps-NE} runs \Cref{Algo:Explore} whose time complexity
 is $O(n^{3})$.  The size of set $N$ is at most $n$, and \Cref{Algo:Ps-NE}
  runs GS-M, whose time complexity is $\Bo{n^{2}}$,
 on $(Q(-w), Q(m;w))$ for each $m \in N$.  Therefore, its running time
 is $\Bo{n^{3}}$ for each woman.  Since there are $n$ women, the total
 running time of \Cref{Algo:Ps-NE} is $\Bo{n^{4}}$.

%%To prove the correctness, suppose that $N' \ne \emptyset$, then it implies that a manipulation strategy was found by\Cref{Algo:Ps-NE}, and therefore $Q$ is not a \Nash. Now we proceed to check for \Ps \Nash.

Suppose that \Cref{Algo:Explore} outputs ``No''.  Then, 
 it implies that a $P$-stable manipulation strategy was found by
 \Cref{Algo:Explore}, and therefore $Q$ is not $P$-stable \Nash.  For the
 opposite direction, suppose that $Q$ is not a $P$-stable \Nash and there
 exists a woman $w$ who has a $P$-stable manipulation strategy $Q'$. % such that $\mu_{Q'}$ is $P$-stable.  
 Then $\mu_{Q'}[w]$ is added to $N$
 when \Cref{Algo:Explore} is run for $w$, and by \Cref{Lemma:unique} the matching $\mu_{Q'}$ is uniquely defined.
 Since $\mu_{Q'}$ is $P$-stable, \Cref{Algo:Ps-NE} must output ``No.''
 \end{proof}

\bibliographystyle{alpha}
\bibliography{refs-GameTheory}

\end{document}